\documentclass[a4paper,envcountsame]{llncs}
\usepackage{\jobname}
\usepackage[utf8]{inputenc}
\usepackage{stmaryrd,mathtools,wrapfig}

\usetikzlibrary{decorations.pathmorphing,decorations.text}
\tikzstyle{kringel}=[decoration={random steps,segment length=3pt,amplitude=1pt,pre=lineto,pre length=.25cm,post=lineto,post length=.25cm}]

\title{Parameterized Complexity of CTL:\\ A Generalization of Courcelle's Theorem}

\author{Martin~Lück\and Arne~Meier \and Irina Schindler\thanks{Supported in part by DFG ME 4279/1-1.}}

\institute{Institut für Theoretische Informatik\\Leibniz Universit\"at Hannover \\ 
\texttt{$\{$lueck, meier, schindler$\}$@thi.uni-hannover.de}}

\begin{document}

\maketitle

\begin{abstract}
We present an almost complete classification of the parameterized complexity of all operator fragments of the satisfiability problem in computation tree logic CTL. The investigated parameterization is the sum of temporal depth and structural pathwidth. The classification shows a dichotomy between W[1]-hard and fixed-parameter tractable fragments. The only real operator fragment which is confirmed to be in FPT is the fragment containing solely AX. Also we prove a generalization of Courcelle's theorem to infinite signatures which will be used to proof the FPT-membership case.
\end{abstract}

\section{Introduction} \label{sect:intro}

Temporal logic is the most important concept in computer science in the area of program verification and is a widely used concept to express specifications. Introduced in the late 1950s by Prior \cite{pr57} a large area of research has been evolved up to today. Here the most seminal contributions have been made by Kripke \cite{kri63}, Pnueli \cite{pn77}, Emerson, Clarke, and Halpern \cite{emha85,clem81} to name only a few. The maybe most important temporal logic so far is the computation tree logic CTL due to its polynomial time solvable model checking problem which influenced the area of program verification significantly. However the satisfiability problem, i.e., the question whether a given specification is consistent, is beyond tractability, i.e., complete for deterministic exponential time. One way to attack this intrinsic hardness is to consider restrictions of the problem by means of operator fragments leading to a trichotomy of computational complexity shown bei Meier \cite{meier11}. This landscape of intractability depicted completeness results for nondeterministic polynomial time, polynomial space, and (of course) deterministic exponential time showing how combinations of operators imply jumps in computational complexity of the corresponding satisfiability fragment.

For more than a decade now there exists a theory which allows us to better understand the structure of intractability: 1999 Downey and Fellows developed the area of parameterized complexity \cite{dofe99} and up to today this field has grown vastly. Informally the main idea is to detect a specific part of the problem, the \emph{parameter}, such that the intractability of the problems complexity vanishes if the parameter is assumed to be constant. Through this approach the notion of \emph{fixed parameter tractability} has been founded. A problem is said to be fixed parameter tractable (or short, FPT) if there exists a deterministic algorithm running in time $f(k)\cdot\textit{poly}(n)$ for all input lengths $n$, corresponding parameter values $k$, and a recursive function $f$. As an example, the usual propositional logic satisfiability problem SAT (well-known to be NP-complete) becomes fixed parameter tractable under the parameter number of variables.

\usetikzlibrary{backgrounds}
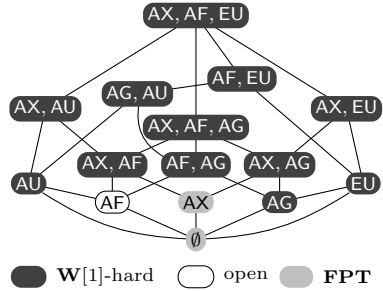
\begin{wrapfigure}[17]{r}{5.3cm}
 \centering
 \begin{tikzpicture}[x=2cm,y=.9cm,bend angle=15,state/.style={rectangle, inner sep=.6mm, rounded corners,font=\scriptsize},scale=0.55,fpt/.style={fill=gray!50!white},w1/.style={fill=gray!50!black,font=\color{white}\scriptsize}]
	\node[state,fpt] (empty) at (0,0) {$\emptyset$};
	\node[state,fpt] (AX) at (0,1) {$\AX$};
	\node[state,fill=white,draw=black] (AF) at (-1,1) {$\AF$};
	\node[state,w1] (AG) at (1,1) {$\AG$};

	\node[state,w1] (AU) at (-2,1.5) {$\AU$};
	\node[state,w1] (EU) at (2,1.5) {$\EU$};
	
	\node[state,w1] (AXAF) at (-1,2) {$\AX,\AF$};
	\node[state,w1] (AFAG) at (0,2) {$\AF,\AG$};
	\node[state,w1] (AXAG) at (1,2) {$\AX,\AG$};
	
	\node[state,w1] (AXAFAG) at (0,3) {$\AX,\AF,\AG$};
	
	\node[state,w1] (AXAU) at (-1.8,3.5) {$\AX,\AU$};
	\node[state,w1] (AXEU) at (1.8,3.5) {$\AX,\EU$};
	
	\node[state,w1] (AGAU) at (-.7,3.9) {$\AG,\AU$};

	\node[state,w1] (AFEU) at (.55,4.3) {$\AF,\EU$};
	
	\node[state,w1] (AXAFEU) at (0,6) {$\AX,\AF,\EU$};

\begin{scope}[on background layer]
	\path[shorten <= -2pt+\pgflinewidth,shorten >= -2pt+\pgflinewidth]
	(empty) edge node {} (AF)
		edge node {} (AX)
		edge[shorten >= -2pt+\pgflinewidth] node {} (AG)
		edge[bend left] node {} (AU)
		edge[bend right] node {} (EU)
	(AF) edge node {} (AU)
		edge node {} (AXAF)
		edge node {} (AFAG)
	(AX) edge node {} (AXAF)
		edge node {} (AXAG)
	(AG) edge node {} (AFAG)
		edge node {} (AXAG)
		edge node {} (EU)
	(AU) edge node {} (AXAU)
		edge node {} (AGAU)
	(EU) edge node {} (AXEU)
		edge node {} (AFEU)
	(AXAF) edge node {} (AXAU)
		edge node {} (AXAFAG)
	(AFAG) edge node {} (AXAFAG)
		edge[bend left,out=35,in=145,shorten <= -2pt+\pgflinewidth] node {} (AGAU)
	(AXAG) edge node {} (AXAFAG)
		edge node {} (AXEU)
	(AXAU) edge node {} (AXAFEU)
	(AXEU) edge node {} (AXAFEU)
	(AXAFAG) edge node {} (AXAFEU)
	(AGAU) edge node {} (AFEU)
	(AFEU) edge node {} (AXAFEU);
\end{scope}

	\node[state,w1,label={0:\scriptsize$\W1$-hard}] at (-2,-1) {$\phantom{\AX}$};
	\node[state,fill=white,draw=black,label={0:\scriptsize open}] at (0,-1) {$\phantom{\AX}$};
	\node[state,fpt,label={0:\scriptsize\FPT}] at (1.2,-1) {$\phantom{\AF}$};
	
\end{tikzpicture}
\caption{Parameterized complexity of $\CTLSAT(\mathcal T)$ parameterized by formula pathwidth and temporal depth (see \Cref{thm:pw+md}).}
\end{wrapfigure}
In this work we almost completely classify the parameterized complexity of all operator fragments of the satisfiability problem for the computation tree logic CTL under the parameterization of formula pathwidth and temporal depth. Only the case for $\AF$ resisted a full classification. We will explain the reasons in the conclusion. For all other fragments we show a dichotomy consisting of two fragments being fixed parameter tractable and the remainder being hard for the complexity class $\W1$ under fpt-reductions. $\W1$ can be seen as an analogue of intractability in the decision case in the parameterized world. To obtain this classification we prove a generalization of Courcelle's theorem \cite{courcelle} for \emph{infinite} signatures which may be of independent interest.

\paragraph{Related work.} Similar research for modal logic has been done by Praveen and influenced the present work in some parts \cite{prav14}. Other applications of Courcelle's theorem have been investigated by Meier~et~al.\ \cite{mstv12} and Gottlob~et~al.\ \cite{gopiwe10}. In 2010 Elberfeld~et~al.\ proved that Courcelle's theorem can be extended to give results in XL as well \cite{ebjata10} wherefore the results of \Cref{cor:AX} can be extended to this class, too.

\section{Preliminaries} \label{sect:prelim}
We assume familiarity with standard notions of complexity theory as Turing machines, reductions, the classes \P and \NP. For an introduction into this field we confer the reader to the very good textbook of Pippenger \cite{pip97b}.

\subsection{Complexity Theory} Let $\Sigma$ be an alphabet. A pair $\Pi = (Q, \kappa)$ is a \emph{parameterized problem} if $Q \subseteq\Sigma^{*}$ and $\kappa \colon \Sigma^{*} \rightarrow \N$ is a function. For a given instance $x\in\Sigma^{*}$ we refer to $x$ as the \emph{input}. A function $\kappa\colon\Sigma^{*}\to\N$ is said to be a \emph{parameterization of $\Pi$} or the \emph{parameter of $\Pi$}. We say a parameterized problem $\Pi$ is \emph{fixed-parameter tractable} (or in the class $\FPT$) if there exists a deterministic algorithm deciding $\Pi$ in time $f(\kappa(x))\cdot|x|^{O(1)}$ for every $x\in\Sigma^{*}$ and a recursive function $f$. Note that the notion of fixed-parameter tractability is easily extended beyond decision problems.

If $\Pi = (Q, \kappa), \Pi' = (Q', \kappa')$ are parameterized problems over alphabets $\Sigma,\Delta$ then an \emph{fpt-reduction from $\Pi$ to $\Pi'$} (or in symbols $\Pi\lefpt\Pi'$) is a mapping $r\colon\Sigma^{*}\to\Delta^{*}$ with the following three properties:
\begin{enumerate}
 \item For all $x\in\Sigma^{*}$ it holds $x\in Q$ iff $r(x)\in Q'$.
 \item $r$ is fixed-parameter tractable, i.e., $r$ is computable in time $f(\kappa(x)) \cdot |x|^\bigO{1}$ for a recursive function $f\colon\N\to\N$.
 \item There exists a recursive function $g\colon\N\to\N$ such that for all $x\in\Sigma^{*}$ it holds $\kappa'(r(x))\leq g(\kappa(x))$.
\end{enumerate}

The class $\W1$ is a parameterized complexity class which plays a similar role as $\NP$ in the sense of intractability in the parameterized world. The class \W1 is a superset of $\FPT$ and a hierarchy of other $\W{}$-classes are build above of it: $\FPT\subseteq\W1\subseteq\W2\subseteq\cdots\subseteq\W\P$. All these classes are closed under fpt-reductions. It is not known whether any of these inclusions is strict. For further information on this topic we refer the reader to the text book of Flum and Grohe \cite{flgr06}.

\subsection{Tree- and Pathwidth} Given a structure $\mathcal A$ we define a \emph{tree decomposition of $\mathcal A$} (with universe $A$) to be a pair $(T,X)$ where $X=\{B_{1},\dots,B_{r}\}$ is a family of subsets of $A$ (the set of \emph{bags}), and $T$ is a tree whose nodes are the bags $B_{i}$ satisfying the following conditions:
\begin{enumerate}
 \item Every element of the universe appears in at least one bag: $\bigcup X=A$.
 \item Every Tuple is contained in a bag: for each $(a_{1},\dots,a_{k})\in R$ where $R$ is a relation in $\mathcal A$, there exists a $B\in X$ such that $\{a_{1},\dots,a_{k}\}\in B$.
 \item For every element $a$ the set of bags containing $a$ is connected: for all $a\in A$ the set $\{B\mid a\in B\}$ forms a connected subtree in $T$.
\end{enumerate}

The \emph{width} of a decomposition $(T,X)$ is $\width(T,X)\dfn\max\{|B|\mid B\in X\}-1$ which is the size of the largest bag minus 1. The \emph{treewidth} of a structure $\mathcal A$ is 
the minimum of the widths of all tree decompositions of $\mathcal A$. Informally the treewidth of a structure describes the tree-likeliness of it. The closer the value is to 1 the more the structure is a tree.

A \emph{path decomposition} of a structure $\mathcal A$ is similarly defined to tree decompositions however $T$ has to be  a path. Here $\pw(\mathcal A)$ denotes the \emph{pathwidth} of $\mathcal A$. Likewise the size of the pathwidth describes the similarity of a structure to a path. Observe that pathwidth bounds treewidth from above.

\subsection{Logic} Let $\Phi$ be a finite set of propositional letters. A \emph{propositional formula} (\PL formula) is inductively defined as follows. The constants $\true,\false,$ (true, false) and any \emph{propositional letter} (or \emph{proposition}) $p\in\Phi$ are \PL formulas. If $\phi,\psi$ are \PL formulas then so are $\phi\land\psi,\lnot\phi,\phi\lor\psi$ with their usual semantics (we further use the shortcuts $\to,\leftrightarrow$). Temporal logic extends propositional logic by introducing four  \emph{temporal operators}, i.e., \emph{next} $\X$, \emph{future} $\F$, \emph{globally} $G$, and \emph{until} $\U$. Together with the two \emph{path quantifiers}, \emph{exists} $\E$ and \emph{all} $\A$, they fix the set of \emph{computation tree logic formulas} (\CTL formulas) as follows. If $\phi\in\PL$ then $\mathsf P\mathsf T\phi,\mathsf P[\phi\U\psi]\in\CTL$ and if $\phi,\psi\in\CTL$ then $\mathsf P\mathsf T\phi,\mathsf P[\phi\U\psi],\phi\lor\psi,\lnot\psi,\phi\land\psi\in\CTL$ hold, where $\mathsf P\in\{\A,\E\}$ is a path quantifier and $\mathsf T\in\{\X,\F,\G\}$ is a temporal operator. The pair of a single path quantifier and a single temporal operator is referred to as a CTL-operator. If $T$ is a set of CTL-operators then $\CTL(T)$ is the restriction of $\CTL$ to formulas that are allowed to use only CTL-operators from $T$.

Let us turn to the notion of Kripke semantics. Let $\Phi$ be a finite set of propositions. A \emph{Kripke structure} $K=(W,R,V)$ is a finite set of \emph{worlds} $W$, a \emph{total successor relation} $R\colon W\to W$ (i.e., for every $w\in W$ there exists a $w'\in W$ with $wRw'$), and an \emph{evaluation function} $V\colon W\to2^{\Phi}$ labeling sets of propositions to worlds. A \emph{path} $\pi$ in a Kripke structure $K=(W,R,V)$ is an infinite sequence of worlds $w_{0},w_{1},\dots$ such that for every $i\in\N$ $w_{i}Rw_{i+1}$. With $\pi(i)$ we refer to the $i$-th world $w_{i}$ in $\pi$. Denote with $\allPaths(w)$ the set of all paths starting at $w$. For $\CTL$ formulas we define the semantics of $\CTL$ formulas $\phi,\psi$ for a given Kripke structure $K=(W,R,V)$, a world $w\in W$, and a path $\pi$ as 
\begin{align*}
 &K,w\models\A\mathsf T\phi &&\Leftrightarrow&&  \text{for all $\pi\in\allPaths(w)$ it holds $K,\pi\models\mathsf T\phi$},\\
 &K,w\models\E\mathsf T\phi &&\Leftrightarrow&& \text{there exists a $\pi\in\allPaths(w)$ it holds $K,\pi\models\mathsf T\phi$},\\
 &K,\pi\models\X\phi &&\Leftrightarrow&& K,\pi(1)\models\phi,\\
 &K,\pi\models\F\phi &&\Leftrightarrow&& \text{there exists an $i\ge0$ such that }K,\pi(i)\models\phi,\\
 &K,\pi\models\G\phi &&\Leftrightarrow&& \text{for all $i\ge0$ }K,\pi(i)\models\phi,\\
 &K,\pi\models\phi\U\psi &&\Leftrightarrow&& \text{$\exists i\ge0\forall j< i$ }K,\pi(j)\models\phi\text{ and }K,\pi(i)\models\psi.
\end{align*}

For a formula $\phi\in\CTL$ we define the satisfiability problem $\CTLSAT$ asking if there exists a Kripke structure $K=(W,R,V)$ and $w\in W$ such that $K,w\models\phi$. Then we also say that $M$ is a \emph{model (of $\phi$}. Similar to before $\CTLSAT(T)$ is the restriction of $\CTLSAT$ to formulas in $\CTL(T)$ for a set of CTL-operators $T$. A formula $\phi\in\CTL$ is said to be in \emph{negation normal form (NNF)} if its negation symbols $\lnot$ occur only in front of propositions; we will use the symbol $\CTLNNF$ to denote the set of CTL-formulas which are in NNF only.

Given $\phi\in\CTL$ we define $\SF\phi$ as the \emph{set of all subformulas of $\phi$} (containing $\phi$ itself). The \emph{temporal depth of $\phi$}, in symbols $\td(\phi)$, is defined inductively as follows. If $\Phi$ is a finite set of propositional symbols and $\phi,\psi\in\CTL$ then 
$$
\begin{array}{lc@{\qquad}ll}
 \td(p) &\dfn 0,& \td(\phi\circ\psi)&\dfn\max\{\td(\phi),\td(\psi)\},\\
 \td(\true) &\dfn 0,& \td(\lnot\phi)&\dfn\td(\phi),\\
 \td(\false) &\dfn 0,& \td(\mathsf P\mathsf T\phi)&\dfn\td(\phi)+1,\\
 && \td(\mathsf P[\phi\U\psi])&\dfn\max\{\td(\phi),\td(\psi)\}+1,\\
\end{array}
$$
where $\circ\in\{\land,\lor,\to,\leftrightarrow\}$, $\mathsf P\in\{\A,\E\}$, and $\mathsf T\in\{\X,\F,\G\}$. If $\psi\in\SF\phi$ then the \emph{temporal depth of $\psi$ in $\phi$} is $\td_{\phi}(\psi)\dfn\td(\phi)-\td(\psi)$.

\emph{Vocabularies} are \emph{finite} sets of \emph{relation symbols} (or \emph{predicates}) of finite arity $k\ge 1$ (if $k=1$ then we say the predicate is \emph{unary}) which are usually denoted with the symbol $\tau$. Later we will also refer to similar objects of infinite size wherefore we prefer to denote them with the term \emph{signature} which usually is an countable infinite sized set of symbols. A \emph{structure} $\mathcal A$ over a vocabulary (or signature) $\tau$ consists of a \emph{universe} $A$ which is a non-empty set, and a relation $P^{\mathcal A}\subseteq A^{k}$ for each predicate $P$ of arity $k$. Monadic second order logic (MSO) is the restriction of second order logic (SO) in which only quantification over unary relations is allowed (elements of the universe can still be quantified existentially or universally). If $P$ is a unary predicate then $P(x)$ is true if and only if $x\in P$ holds (otherwise it is false).

\section{Parameterized Complexity of CTL-SAT(T)}\label{pg:structure}
In this section we investigate all operator fragments of $\CTLSAT$ parameterized by temporal depth and formula pathwidth with respect to its parameterized complexity. This means, we the given formulas from $\CTL$ as input are represented by relational structures as follows.

Let $\varphi\in\CTL$ be a \CTL formula. The vocabulary of our interest is $\tau$ being defined as $\tau\dfn \{\const^1_f \mid f\in\{\true,\false\}\} \, \cup \, \{\conn_{f,i}^2 \mid f \in\{\land,\lor,\lnot\}, 1 \leq i \leq \arity f\} \, \cup \, \{\var^{1}, \repr^{1}, \reprPL^{1}\} \,\cup\,\{\repr_\mathsf{C}^{1}, \body_\mathsf{C}^{2}\mid\mathsf C\text{ is a unary CTL-operator}\}\,\cup\,\{\repr_\mathsf{C}^{1}, \body_\mathsf{C}^{3}\mid\mathsf C\text{ is a binary CTL-operator}\}$. We then associate the vocabulary $\tau$ with the structure $\mathcal{A}_{\varphi}$ where its universe consists of elements representing subformulas of $\varphi$. The predicates are defined as follows
\begin{itemize}
 \item $\var^{1}(x)$ holds iff $x$ represents a variable,
 \item $\repr^{1}(x)$ holds iff $x$ represents the formula $\varphi$,
 \item $\reprPL^{1}(x)$ holds iff $x$ represents a propositional formula,
 \item $\repr_{\mathsf{C}}^{1}(x)$ holds iff $x$ represents a formula $\mathsf C\psi$ where $\mathsf C$ is a CTL-operator,
 \item $\body_{\mathsf C}^{2}(y,x)$ holds iff $x$ represents a formula $\mathsf C\psi$ and $\psi$ is represented by $y$ where $\mathsf C$ is a unary CTL-operator,
 \item $\body_{\mathsf C}^{3}(y,z,x)$ holds iff $x$ represents a formula $\mathsf C(\psi,\chi)$ and $\psi$ / $\chi$ is represented by $y$ / $z$ where $\mathsf C$ is a binary CTL-operator,
 \item $\const_f^1(x)$ holds iff $x$ represents the constant of $f$,
 \item $\conn^2_{f,i}(x,y)$ holds iff $x$ represents the $i$th argument of the function $f$ at the root of the formula tree represented by $y$.
 \end{itemize}

As an example, the corresponding structure $\mathcal A_\varphi$ for the formula
 $
  \varphi\dfn\EX(\AG(p\land\lnot (\EF z)))\lor(\lnot(\A[p\U(\EF z)]))
 $
is shown in \Cref{fig:Aphi}.

\begin{figure}
 \centering
  \begin{tikzpicture}[every node/.style={font=\sffamily}, state/.style={circle,fill=darkgray,minimum width=2mm,inner sep=0mm},y=1.3cm,x=1.5cm]
 \node[state,label={0:$\varphi$},label={180:\footnotesize repr}] (phi) at (0,0) {};
 
 \node[state,label={180:$\EX(\AG(p\land\lnot (\EF z)))$},label={0:\footnotesize repr$_{\text{EX}}$}] (EX) at (-1,-1) {};
\node[state,label={0:$\lnot(\A[p\U(\EF z)])$}] (neg) at (1,-1) {};

 \node[state,label={180:$\AG(p\land\lnot (\EF z)))$},label={0:\footnotesize repr$_{\text{AG}}$}] (AG) at (-1,-2) {};
 \node[state,label={0:$\quad\A[p\U(\EF z)]$},label={180:\footnotesize repr$_{\text{AU}}$}] (AU) at (1,-2) {};
 
 \node[state,label={0:$ p$},label={180:$\substack{\text{\footnotesize repr$_{\text{PL}}$}\\\text{var}}$}] (p) at (.3,-3) {};
 \node[state,label={0:$\EF z$},label={180:\footnotesize repr$_{\text{EF}}$}] (EF1) at (1.7,-3) {};
 
 \node[state,label={0:$ z$},label={180:$\substack{\text{\footnotesize repr$_{\text{PL}}$}\\\text{var}}$}] (z) at (1.7,-4) {};
 
 \node[state,label={180:$ p\land\lnot(\EF z)$},label={0:}] (pAnd-) at (-1.7,-3) {};
 
\node[state,label={180:$ \lnot(\EF z)$},label={0:}] (-EF) at (-2.4,-4) {};
 
 \node[state,label={180:$\EF z$},label={0:\footnotesize repr$_{\text{EF}}$}] (EF2) at (-2.4,-5) {};

\path[-,darkgray] (phi) edge node [midway, above, sloped,outer sep=-1mm] {\footnotesize conn$_{\lor,1}$} (EX);
\path[-,darkgray] (phi) edge node [midway, above, sloped,outer sep=-1mm] {\footnotesize conn$_{\lor,2}$} (neg);
\path[-,darkgray] (EX) edge node [midway, above, sloped,outer sep=-1mm] {\footnotesize body$_{\text{EX}}$} (AG);
\path[-,darkgray] (neg) edge node [midway, above, sloped,outer sep=-1mm] {\footnotesize conn$_{\lnot,1}$} (AU);
\path[-,darkgray] (AU) edge node [midway, above, sloped,outer sep=-1mm] {\footnotesize body$_{\text{AU},1}\;$} (p);
\path[-,darkgray] (AU) edge node [midway, above, sloped,outer sep=-1mm] {\footnotesize body$_{\text{AU},2}$} (EF1);
\path[-,darkgray] (EF1) edge node [midway, above, sloped,outer sep=-1mm] {\footnotesize body$_{\text{EF}}$} (z);
\path[-,darkgray] (AG) edge node [midway, above, sloped,outer sep=-1mm] {\footnotesize body$_{\text{AG}}\;$} (pAnd-);
\path[-,darkgray] (pAnd-) edge [out=315,in=225] node [midway, above, sloped,outer sep=-.8mm] {\footnotesize conn$_{\land,1}$} (p);
\path[-,darkgray] (pAnd-) edge node [midway, above, sloped,outer sep=-1mm] {\footnotesize conn$_{\land,2}\;$} (-EF);
\path[-,darkgray] (-EF) edge node [midway, above, sloped,outer sep=-1mm] {\footnotesize conn$_{\lnot,1}\;$} (EF2);
\path[-,darkgray] (EF2) edge [out=335,in=235] node [midway, above, sloped,outer sep=-.8mm] {\footnotesize body$_{\text{EF}}$} (z);
\end{tikzpicture}
 \caption{Example relational structure $\mathcal A_\varphi$.}\label{fig:Aphi}
\end{figure}
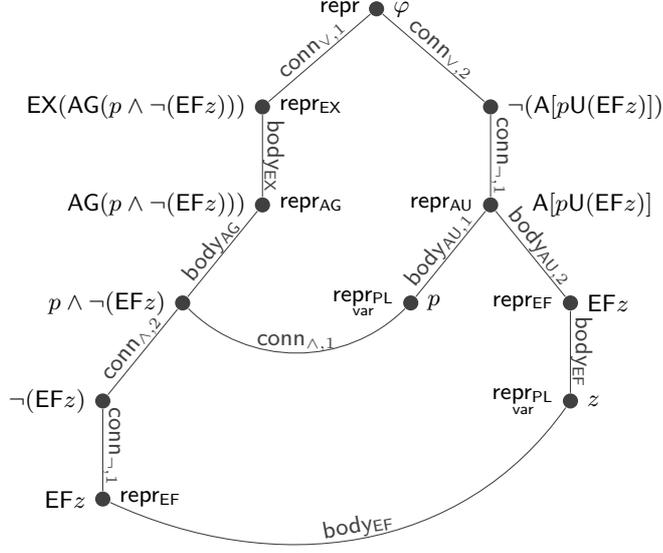

Now we consider the problem $\CTLSAT$ parameterized by the pathwidth of its instance structures $\mathcal A_\varphi$ (for the instances $\varphi$) as well as the temporal depth of the formula. Hence the parameterization function $\kappa$ maps, given an instance formula $\varphi\in\CTL$ to the pathwidth of the structures $\mathcal A_\varphi$ plus the temporal depth of $\varphi$, i.e., $\kappa(\varphi)=\pw(\mathcal A_\varphi)+\td(\varphi)$.

The following theorem summarizes the collection of results we have proven in the upcoming lemmas. The subsection on page~\pageref{pg:genCourcelle} contains the \FPT result together with the generalization of Courcelle's theorem to infinite signatures.

\begin{theorem}\label{thm:pw+md}
 $\CTLSAT(T)$ parameterized by formula pathwidth and temporal depth is
\begin{enumerate}
 \item in $\FPT$ if $T=\{\AX\}$ or $T=\emptyset$, and 
 \item $\W1$-hard if $\AG \in T$, or $\AU\in T$, or $\{\AX,\AF\} \subseteq T$.
\end{enumerate}
\end{theorem}

\begin{proof}(1.) is witnessed by \Cref{cor:AX}. The proof of (2.) is split into \Cref{lem:pw+md+AU,lem:pw+md+AG,lem:AXAF-W1}. \qed
\end{proof}

One way to prove the containment of a problem parameterized in that way in the class \FPT is to use the prominent result of Courcelle \cite[Thm.~6.3 (1)]{courcelle}. Informally, satisfiability of CTL-formulas therefore has to be formalized in monadic second order logic. The other ingredient of this approach is expressing formulas by relational structures as described before. Now the crux is that our case requires a family of MSO formulas which depend on the instance. This however seems to be a serious issue at first sight as this prohibits the application of Courcelle's theorem. Fortunately we are able to generalize Courcelle's theorem in a way to circumvent this problem. Moreover we extended it to work with infinite sized signatures under specific restrictions which allows us to state the desired \FPT result described as follows.

\subsection*{A Generalized Version of Courcelle's Theorem}\label{pg:genCourcelle}

Assume we are able to express a problem $Q$ in MSO. If instances $x\in Q$ can be modeled via some relational structure $\mathcal A_{x}$ over some finite vocabulary $\tau$ and we see $Q$ as a parameterized problem $(Q,\kappa)$ where $\kappa$ is the treewidth of $\mathcal A_{x}$ then by Courcelle's theorem we immediately obtain that $(Q,\kappa)$ is in $\FPT$ \cite{courcelle}. If we do not have a fixed MSO formula (which is independent of the instance) then we are not able to use the mentioned result. However the following theorem shows how it is possible even with infinite signatures  to apply the result of Courcelle. For this, we assume that the problem can be expressed by an infinite family $(\phi_{n})_{n\in\N}$ of MSO-formulas along with the restriction that $(\phi_{n})_{n\in\N}$ is uniform, i.e., there is a recursive 
function $f \colon n \to \phi_{n}$.

Let $\kappa$ be a parameterization. Call a function $f : \Sigma^* \to \Sigma^*$ $\kappa$\emph{-bounded} if there is a computable
function $h$ such that for all $x$ it holds that $\size{f(x)} \leq h(\kappa(x))$. 

\begin{theorem}\label{thm:genCourcelle}
 Let $(Q,\kappa)$ be a parameterized problem such that instances $x\in\Sigma^{*}$ can be expressed via relational structures $\mathcal A_{x}$ over a (possibly infinite) signature $\tau$ and $\tw(\mathcal{A}_x)$ is $\kappa$-bounded.
If there exists a uniform MSO-formula family $(\phi_{n})_{n\in\N}$ and a fpt-computable, $\kappa$-bounded function $f$ such that for all $x\in\Sigma^{*}$ it holds $x \in Q \Leftrightarrow \mathcal A_{x}\models\phi_{\size{f(x)}}$ then $(Q,\kappa)\in\FPT$.
\end{theorem}
\begin{proof}
Let $(Q,\kappa)$, $(\phi_{n})_{n\in\N}$, $\kappa$ and $f$ be given as in the conditions of the theorem. Let
$(\phi_{n})_{n\in\N}$ be computed by a w.l.o.g.\ non-decreasing and computable function $g$.
The following algorithm correctly decides $Q$ in fpt-time w.r.t. $\kappa$.
First compute $i := \size{f(x)}$ in \FPT for the given instance $x$. Since $(\phi_{n})_{n\in\N}$ is uniform and $f$ is $\kappa$-bounded we can construct $\phi_i$ in time $g(n) = g(|f(x)|) \leq g(h(\kappa(x)))$ for recursive $g$, hence in \FPT.
Now we are able to solve the model checking problem instance $(\mathcal{A}_x, \phi_i)$ in time $f'\big(\tw(\mathcal A_{x}), \size{\phi_i}\big) \cdot \size{\mathcal{A}_x}$ for a recursive $f'$ due to Courcelle's theorem. As both $\tw$ and $\size{\phi_i}$ are $\kappa$-bounded, the given algorithm then runs in \FPT time.
 \qed
\end{proof}


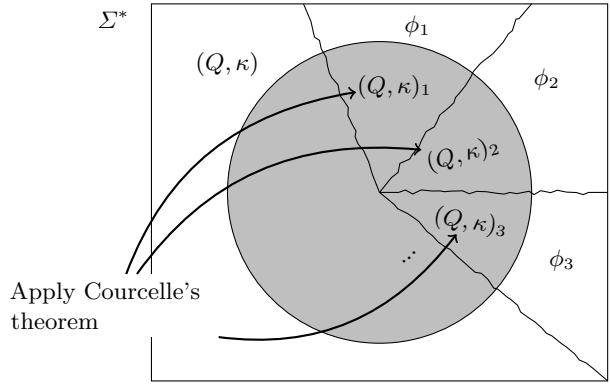
\begin{figure}
 \centering
 \pgfmathsetseed{1342}
\begin{tikzpicture}[]
 \draw[black,fill=white] (0,0) -- (6,0) -- (6,5) -- (0,5) -- cycle;
 \node at (-0.5,4.8) {$\Sigma^{*}$};
 
 \path[draw,fill=lightgray] (5,2.5) arc (0:360:2);

 \draw[black] (3,2.5) decorate [kringel] {-- (2,5)};
 \draw[black] (3,2.5) decorate [kringel] {-- (5,5)};
 \draw[black] (3,2.5) decorate [kringel] {-- (6,2.5)};
 \draw[black] (3,2.5) decorate [kringel] {-- (6,0)};

\path[postaction={decorate,decoration={text along path,text align=center,text={...}}}] (3,1.5) to [bend right=45]  (3.6,2);

 \node at (3.5,4.7) {$\phi_{1}$};
 \node at (5.2,4) {$\phi_{2}$};
 \node at (5.4,1.6) {$\phi_{3}$};

 \node at (1,4.2) {$(Q,\kappa)$};

 \node (q1) at (3.2,3.9) {$(Q,\kappa)_{1}$};
 \node[rotate=15] (q2) at (4.1,3) {$(Q,\kappa)_{2}$};
 \node[rotate=-10] (q3) at (4.2,2.1) {$(Q,\kappa)_{3}$};
 
 \node[fill=white,text width=2.7cm] (app) at (-0.5,1) {Apply Courcelle's theorem};
 
 \path[->, thick, draw=black,shorten >=-1mm] 
 (app) edge [bend left] node {} (q1)
 edge [bend left] node {} (q2)
 edge [bend right] node {} (q3);
\end{tikzpicture}
\caption{Visualization of the infinite application of Courcelle's theorem in \Cref{thm:genCourcelle}. $(Q,\kappa)_i$ for $i\in\N$ are the slices of the parameterized problem, i.e., $(Q,\kappa)_i\dfn\{x\in\Sigma^*\mid x\in Q\text{ and } \kappa(x)=i\}$.}\label{fig:genCourcelle}
\end{figure}

Note that the infinitely sized signature is required to describe the structures from the set of all structures $\mathfrak A$ which occur with respect to the corresponding \emph{family} of MSO-formulas $(\phi_n)_{n\in\N}$. For every subset $T\subset\mathfrak A$ of structures with respect to each $\phi_i$ then have (as desired and required by Courcelle's theorem) a finite signature, i.e., a vocabulary.\medskip

Praveen \cite{prav14} shows the fixed-parameter tractability of $\MLSAT$ (parameterized by pathwidth and modal depth) by applying Courcelle's theorem, using for each modal formula an MSO-formula whose length is linear in the modal depth. This can be seen as a special case of \Cref{thm:genCourcelle} using a \P-uniform MSO family that partitions the instance set according to the modal depth.\bigskip


Again we want to stress that formula pathwidth of $\varphi$ refers to the pathwidth of the corresponding structures $\mathcal A_\varphi$ as defined above. 

\begin{lemma}\label{lem:AX}
Let $\varphi\in\CTLNNF(\{\AX,\EX\},B)$ given by the structure $\mathcal A_{\varphi}$ over $\tau$. Then there exists an MSO formula $\theta(\varphi)$ such that $\varphi\in\CTLSAT(\{\AX\})$ iff $\mathcal A_\varphi\models\theta(\varphi)$ and $\theta(\varphi)$ depends only on $\td(\varphi)$.
\end{lemma}

\begin{proof}
The first step is to show that a formula $\varphi \in \CTLNNF(\{\AX,\EX\})$ is satisfiable if and only if it is satisfied by a Kripke structure of depth $\td(\varphi)$, where the depth of a structure $(M, w_0)$ is the maximal distance in $M$ from $w_0$ to another state from $M$. This can be similar proven as the \emph{tree model property} of modal logic \cite[p. 269, Lemma~35]{blrive01}.

Let $\varphi$ be the given formula in $\CTLNNF(\{\AX,\EX\})$. The following formula $\theta_{\textit{struc}}$ describes the properties of the structure $\mathcal A_{\varphi}$. At first it takes care of the uniqueness of the formula representative. If an element $x$ does not represent a formula then it has to be a subformula. Additionally if $x$ it is not a variable it has to be either a constant, or a Boolean function $f\in B$ with the corresponding arity $\arity f$, or an $\AX$-, or an $\EX$-formula respectively. Furthermore the distinctness of the representatives has to be ensured which together with the previous constraints implies acyclicity.

In the following $f_1(u,v,w,x)$ corresponds to the operator of the function which is true if exactly one of its arguments is true.
{\scriptsize
\begin{align*}
 \theta_{\text{struc}} \dfn & 
 \forall x\forall y(\repr(x) \land \repr(y) \to x=y)\land\\
&\forall x \Bigg( \neg \repr(x) \to \exists y \Big( \neg \var(y) \land  \bigvee_{\stackrel{f \in \{\land,\lor,\lnot\},}{1\leq i\leq\arity{f}}} \conn_{f,i}(x,y)\Big)\Bigg)\land \\ 
& \forall x\, f_1\Bigg( \var(x), \bigvee_{f \in \{\true,\false\}} \const_{f}(x), \\
&\quad\qquad \bigvee_{\stackrel{f \in B,}{\arity{f} \geq 1}} \bigwedge_{1\leq i\leq\arity{f}} \exists y  \big( \conn_{f,i}(y,x)\land \forall z \big( \conn_{f,i}(z,x) \to z=y \big) \big),\\
&\quad\qquad \exists y \big( \bodyAX(y,x) \land \forall z \big( \bodyAX (z,x) \to z=y \big) \big),\\
&\quad\qquad \exists y \big( \bodyEX(y,x) \land \forall z \big( \bodyEX (z,x) \to z=y \big) \big)\Bigg) \land\\
&\forall x\forall y \big((\bodyAX(y,x)\to\reprAX(x))\land(\bodyEX(y,x)\to\reprEX(x))\big).
\end{align*}
}
The previous formula is a modification of the formula used in the proof of Lemma~1 in \cite{mstv12}.

The next formulas will quantify sets $M_i$ which represent sets of satisfied subformulas at worlds in the Kripke structure at depth $i$. Here the formulas with propositional connectives, resp., all constants,  have a valid assignment obeying their function value in the model $M_i$. The $\AX$- and $\EX$-formulas are processed as expected: the $\EX$-formulas branch to different worlds and the $\AX$-formulas have to hold in all possible next worlds. Now we are ready to define $\theta_{\text{assign}}^{i}$ in an inductive way. At depth $0$ we want to consider only propositional formulas. Here it ensures that all Boolean functions obey the model:
{\scriptsize
\begin{align*}
\theta^{0}_{\text{assign}}(M_{0}) \dfn \forall x, y_{1},&\dots, y_{n}\in M_{0}:\reprPL(x)\land\\
 \bigwedge_{f\in B}\Biggl(\quad&\bigwedge_{\mathclap{\arity{f}=0}}\const_{f}(x)\to f\land
 \bigwedge_{\mathclap{1\leq i\leq\arity{f}}}\conn_{f,i}(y_{i},x)\to f(M_{0}(y_{1}),\dots,M_{0}(y_{\arity{f}}))\Biggr).
\end{align*}
}
In the general definition of $\theta_{\textit{assign}}^{i}$ we utilize for convenience two subformulas, $\theta^{i}_{\text{branchEX}}$ and $\theta^{i}_{\text{stepAX}}$. The first is defined for an element $x$ representing an $\EX$-formula, a set of elements $M_{i}$ representing to be satisfied formulas, and a set of elements $M_{\AX}$ representing the $\AX$-formulas which are satisfied in the current world. The formula enforces that the formula $\EX\psi$ represented by $x$ has to hold in the next world together with all bodies of the $\AX$-formulas:
{\scriptsize
\begin{align*}
 \theta^{i}_{\text{branchEX}}(M_{i},M_{\AX},x) \dfn \exists y\Big(&\bodyEX(y,x)\land\\
 &\exists M_{i-1}\big(M_{i-1}(y) \land \forall z\in M_{\AX}(\exists w \, \bodyAX(w, z) \land M_{i-1}(w))\land\\
 &\phantom{\exists M_{i-1}\big(}\theta^{i-1}_{\text{assign}}(M_{i-1})\big)\Big).
\end{align*}
}
The second formula is crucial when there are no $\EX$-formulas represented in $M_{i}$. Then the $\AX$-formulas still have to be satisfied eventually wherefore we proceed with a single next world (without any branching required):
{\scriptsize
\begin{align*}
 \theta^{i}_{\text{stepAX}}(M_{\AX}) \dfn \exists M_{i-1}\forall z\in M_{\AX}(\exists w \, \bodyAX(w, z) \land M_{i-1}(w))\land\theta^{i-1}_{\text{assign}}(M_{i-1}).
\end{align*}
}%
Now we turn towards the complete inductive definition step where we need to differentiate between the two possible cases for representatives: either a propositional or a temporal formula is represented. The first part is similar to the induction start and the latter follows the observation that for every $\EX$-preceded formula we want to branch. In each such branch all not yet satisfied $\AX$-preceded formulas have to hold. The set $M_{\AX}$ contains all $\AX$-formulas which are satisfied in the current world. 
If we do not have any $\EX$-formulas then we enforce a single next world for the remaining $\AX$-formulas.
{\scriptsize
\begin{align*}
 \theta^i_{\text{assign}}(M_{i}) \dfn &\forall x,y_{1},\dots, y_{n}\in M_{i}\\ 
 &\bigwedge_{f\in B}\Biggl(\quad\bigwedge_{\mathclap{\arity{f}=0}}\const_{f}(x)\to (M_{i}(x)\leftrightarrow f)\land\\\displaybreak
 &\qquad\quad\;\,\bigwedge_{\mathclap{1\leq i\leq\arity{f}}}\conn_{f,i}(y_{i},x)\to\big(M_{i}(x)\leftrightarrow f(M_{i}(y_{1}),\dots, M_{i}(y_{\arity{f}}))\big)\Bigg)\land\\
 &\exists M_{\AX}\subseteq M_{i}\Big(\forall x\Big(M_{\AX}(x) \leftrightarrow \big(\reprAX(x)\land M_{i}(x) \big)\Big)\land\\
 &\phantom{\exists M_{\AX}\subseteq M_{i}\Big(}\forall x\in M_{i}\big(\reprEX(x)\to\theta^{i}_{\text{branchEX}}(M_{i},M_{\AX},x)\big)\land\\
 &\phantom{\exists M_{\AX}\subseteq M_{i}\Big(}\big(\forall x\in M_{i}(\lnot\reprEX(x))\big)\to \theta^{i}_{\text{stepAX}}(M_{\AF}\big)
 \Big)
\end{align*}}

Through the construction we get that $\varphi$ is satisfiable iff 
$\mathcal A_{\varphi}\models\theta_{\text{struc}}\land\exists M(\theta^{\td(\varphi)}_{\text{assign}}(M))\nfd\theta(\varphi)$. \qed
\end{proof}

\begin{corollary}\label{cor:AX}
 $\CTLSAT(\{\AX\})$ parameterized by formula pathwidth and temporal depth is fixed-parameter tractable.
\end{corollary}
\begin{proof}
Assume that the given formula $\varphi$ is in NNF since such a transformation is possible in linear time.
As pathwidth is an upper bound for treewidth, we apply \Cref{thm:genCourcelle} in the following way. For $\size{f(\varphi)} = \td(\varphi)$ the function $f$ is $\kappa$-bounded and computes
the appropriate MSO formula from the uniform family given by \Cref{lem:AX}.
\qed
\end{proof}

\subsection*{Intractable fragments of CTL-SAT}\label{pg:intracCTL}
In the following section we consider fragments of CTL for which their models cannot be bounded by the temporal depth of the formula. Therefore the framework used for the $\AX$ case cannot be applied. Instead we prove $\W{1}$-hardness.

\begin{lemma}\label{lem:AXAG-W1}
  $\CTLSAT(T)$ parameterized by formula pathwidth and temporal depth is $\W1$-hard if $\{\AX,\AG\} \subseteq T$.
\end{lemma}
\begin{proof}
 We will modify the construction in the proof of Praveen \cite[Lemma A.3]{prav14} and thereby state an fpt-reduction from the parameterized problem p-PW-SAT whose input is ($\mathcal{F}$, $\textit{part}: \Phi \to [k]$, $\textit{tg}: [k] \to \mathbb{N}$), where $\mathcal F$ is a propositional CNF formula, $\textit{part}$ is a function that partitions the set of propositional variables of $\mathcal F$ into $k$ parts, and $\textit{tg}$ is a function which maps to each part a natural number. The task is to find a satisfying assignment of $\mathcal F$ such that in each part $p\in[k]$ exactly $\textit{tg}(p)$ variables are set to true. 
 A generalization of this problem to arbitrary formulas $\mathcal F$ (i.e., the CNF constraint is dropped) is $\W1$-hard when parameterized by $k$ and the pathwidth of the structural representation $\mathcal A_{\mathcal F}$ of $\mathcal F$ which is similar proven as in \cite[Lemma 7.1]{prav14}. 
 
 The further idea is to construct a $\CTL$-formula $\phi_{\mathcal{F}}$ in which we are able to verify the required targets. The formula enforces a Kripke structure $K=(W,R,V)$ where in each world $w\in W$ the value of $V(w)$ coincides with a satisfying assignment $f$ of $\mathcal{F}$ together with the required targets. Each such $K$ contains as a substructure a chain $w_{0} R w_{1} R \cdots R w_{n}$ of worlds and all variables $q_{i}$ in $\mathcal{F}$ are labeled to each $w_{j}$ if $f(q_{i})$ holds.

Let $q_{1}, \dots, q_{n}$ be all the propositional variables in $\mathcal{F}$. Then $t_{\uparrow1}, \dots , t_{\uparrow k},$ respectively, $f_{\uparrow1}, \dots, f_{\uparrow k}$ are propositions to distinguish the parts, $tr_{p}^{0}, \dots,tr_{p}^{n[p]},$ respectively, $fl_{p}^{0}, \dots, fl_{p}^{n[p]}$ for $p \in [k],$ are counter propositions for the number of variables set to true and false, $d_{0}, \dots, d_{n+1}$ are depth propositions, and $\Phi(p)$ denotes the set of variables in part $p\in[k]$.

The formula $\phi_{\mathcal F}$ that is the conjunction of subformulas (\Cref{fig:ax-ag}) similar to \cite[Lemma A.3]{prav14} states the reduction from p-PW-SAT to $\CTLSAT(\{\AX,\AG\})$ parameterized by temporal depth and pathwidth.

\begin{figure}[t]
{\scriptsize
 \begin{align*}
 \textit{determined} \dfn & \AG\bigwedge_{i=1}^{n}\Big((q_i\Rightarrow\AX q_{i})\land (\neg q_i\Rightarrow\AX \neg q_{i})\Big)\\
 \textit{depth} \dfn & \bigwedge_{i=0}^{n}\Big((d_{i}\land \neg d_{i+1})\Rightarrow \AX(d_{i+1} \land \neg d_{i+2})\Big)\\
 \textit{setCounter} \dfn & \AG\bigwedge_{i=1}^{n}\Big((d_{i}\land \neg d_{i+1})\Rightarrow \left[(q_{i}\Rightarrow t_{\uparrow part(i)})\land (\neg q_{i}\Rightarrow f_{\uparrow part(i)})\right]\Big)\\
 \textit{incCounter} \dfn & \AG\bigwedge_{p=1}^{k}\bigwedge_{j=0}^{n[p]}\Big[ \Big(t_{\uparrow p} \Rightarrow \left(tr^{j}_{p} \Rightarrow \AX tr^{j+1}_{p}\right)\Big) \land \Big(f_{\uparrow p} \Rightarrow\left(fl^{j}_{p} \Rightarrow \AX fl^{j+1}_{p}\right)\Big)\Big] \\
 \textit{targetMet} \dfn & \AG\bigwedge_{p=1}^k (d_{n+1}\Rightarrow (tr_p^{tg(p)}\land\lnot tr_p^{tg(p)+1}\land fl_p^{n[p]-tg(p)}\land\lnot tr_p^{n[p]-tg(p)+1}))\\
 \textit{countInit} \dfn & d_0 \land \neg d_1 \land \AG \bigwedge_{p=1}^{k}\left( tr_{p}^{0} \land  fl_{p}^{0} \right)\\
 \textit{countMonotone}_1 \dfn & \AG\bigwedge_{p=1}^{k} \bigwedge_{j=0}^{n[p]} \Big[\Big( tr_{p}^{j} \Rightarrow \AG tr_{p}^{j})\Big)\land \Big( fl_{p}^{j} \Rightarrow \AG fl_{p}^{j}\Big)\Big]\\
  \textit{countMonotone}_2 \dfn & \AG\left(\bigwedge_{i=1}^{n} \Big((d_i \Rightarrow d_{i-1}) \Big)\land \bigwedge_{p=1}^{k} \bigwedge_{l=2}^{n[p]}\Big[ (tr_{p}^{j}\Rightarrow tr_{p}^{j-1}) \land (fl_{p}^{j}\Rightarrow fl_{p}^{j-1})\Big]\right)\\
\end{align*}
}
\caption{Reduction from p-PW-SAT to $\CTLSAT(\{\AX,\AG\})$\label{fig:ax-ag}}
\end{figure}

 In the following we assume the chain of worlds as explained before to be the relevant part of the model. The world where the conjunction $\phi_{\mathcal F}$ holds is assumed to be $w_{0}$. The formula \emph{determined} forces the variables $q_{i}$ not to change their value in successor levels by passing  the value of each $q_{i}$ to all next levels. Hence we get $\mathcal{M},w_{0} \models \mathcal{F}\land \textit{determined}$. \emph{depth} ensures that in the world $w_{i}$ holds $d_{i}\land \lnot d_{i+1}$, $\mathcal{M},w_{0} \models \textit{determined} \land (d_{0}\land \lnot d_{1})$ by \emph{countInit}. In the next formula \emph{setCounter} the variable $t_{\uparrow part(i)}$ holds if $q_{i}$ is set to true at the world $w_{i}$, respectively, the variable $f_{\uparrow part(i)}$ if $q_{i}$ does not hold at $w_{i}$. 
 
 Now we use the variables $t_{\uparrow part(i)}$ to increment the counter propositions $tr_{p}^{0}, \dots, tr_{p}^{n[p]}$ for all variables set to true in the formula \textit{incCounter} as follows. If the $j$ variables $\Phi(p)\cap \{q_{1}, \dots q_{j}\}$ at $\textit{part}(p)$ are set to true so is $t_{\uparrow p}$ set at the world $w_{i}$ to true and all successors of $w_{i}$ force increment of the value $\ell$ in $tr^{\ell}_{p}$, respectively, $fl^{\ell}_{p}$. This is ensured stepwise depending on the temporal depth $n$. 
 
The counters for the target function $tr_{p}^{i}$ and $fl_{p}^{i}$ are initialized by \emph{countInit}, i.e., $tr_{p}^{0}$ and $fl_{p}^{0}$ are set to true in all $w_{i}$s. Additionally the \textit{countMonotone} formulas requires the counter values to be stable and nondecreasing. The given target function $tg:[k]\to \mathbb{N}$ is then checked with the formula \textit{targetMet} such that $\mathcal{M},w_{0} \models \textit{targetMet}$, i.e., in the world at depth $n+1$ the target proposition $tr_{p}^{\textit{tg}(p)}$ must hold (and must stop, i.e., $tr_{p}^{\textit{tg}(p)+1}$ is false) for each part $p\in[k]$.  The correctness of the reduction is similarly proven as in \cite[Lemma A.3]{prav14}. \qed
\end{proof}

\begin{lemma}\label{lem:AXAF-W1}
  $\CTLSAT(T)$ parameterized by formula pathwidth and temporal depth is $\W1$-hard if $\{\AX,\AF\} \subseteq T$.
\end{lemma}
\begin{proof}
The formulas in the proof of \Cref{lem:AXAG-W1} are deliberately chosen to have $\AG$ operators only at temporal depth zero and in conjunctions, leading to a formula $\phi_{\mathcal F}=\psi\land\AG\chi$, where $\psi$ is purely propositional and $\chi\in\CTL(\{\AX\})$. 
As $\AG(\alpha)\land\AG(\beta)\equiv\AG(\alpha\land\beta)$ we can modify the formula $\phi_{\mathcal F}$ which is a conjunction of the formulas from above to the desired form containing only a single $\AG$. Therefore $\AG$ can be replaced by $\EG$ and the proof stays valid since there is only one instance of an existential temporal operator and it occurs at temporal depth zero, and only one path of the model is required to form a chain of worlds which express the weighted partitioned satisfiability.  \qed
\end{proof}

\begin{lemma}\label{lem:pw+md+AG}
 $\CTLSAT(T)$ parameterized by formula pathwidth and temporal depth is $\W1$-hard if $\AG\in T$.
\end{lemma}
\begin{proof}
Now we consider the case were $T=\{ \AG \}$. Both the operators $\AG$ and $\EF$ are \emph{stutter-invariant}, \ie, they cannot distinguish a path $\pi$ and a path $\pi'$ which is obtained from $\pi$ by duplicating arbitrary worlds on the path. Hence a more sophisticated construction is required to maintain correctness of the reduction.
The first problem is that even if we can enforce the worlds $w_1,\ldots,w_n$ to appear, we cannot avoid redundant intermediate worlds.
Therefore we cannot check the counter for its exact value but only bound its value from above. Instead we demand upper bounds
for both the number of ones and the number of zeros in a partition.
The second problem is that we cannot express without $\X$ that the counter has to increase in the next world (but not in the current). To circumvent this we label together with the depth propositions their ``parities''. This requires two variables and does not increase the pathwidth much. The following formulas from \Cref{lem:AXAG-W1} have to be changed:
{\scriptsize
 \begin{align*}
  \textit{determined} \dfn & \bigwedge_{i=1}^{n}\Big((q_i\Rightarrow\AG q_{i})\land (\neg q_i\Rightarrow\AG \neg q_{i})\Big)\\
 \textit{depth} \dfn & \AG \bigwedge_{i=0}^{n}\Big((d_{i}\land \neg d_{i+1})\Rightarrow (m_{i \bmod 2} \land \neg m_{1 - (i \bmod 2)} \land \EF(d_{i+1} \land \neg d_{i+2}))\Big)\\
 \textit{setCounter} \dfn & \AG\bigwedge_{i=1}^{n}\Big((d_{i}\land \neg d_{i+1})\Rightarrow \left[(q_{i}\Rightarrow t_{\uparrow part(i)})\land (\neg q_{i}\Rightarrow f_{\uparrow part(i)})\right]\Big)\\
 \textit{incCounter} \dfn & \AG\bigwedge_{p=1}^{k}\bigwedge_{j=0}^{n[p]-1}\bigwedge_{m=0}^{1}\Big[ \Big((t_{\uparrow p} \land tr^{j}_{p} \land m_i) \Rightarrow \AG\left( m_{1-i} \Rightarrow \AG tr^{j+1}_{p}\right)\Big)\\
 & \hphantom{\AG\bigwedge_{p=1}^{k}\bigwedge_{j=0}^{n[p]-1}\bigwedge_{m=0}^{1}\Big[} \land \Big(f_{\uparrow p} \land fl^{j}_{p} \land m_i \Rightarrow \AG\left(m_{1-i} \Rightarrow \AG fl^{j+1}_{p}\right)\Big)\Big]
\end{align*}}
\qed
\end{proof}

\begin{lemma}\label{lem:pw+md+AU}
 $\CTLSAT(T)$ parameterized by formula pathwidth and temporal depth is $\W1$-hard if $\AU\in T$.
\end{lemma}
\begin{proof}
We further modify the reduction from \Cref{lem:pw+md+AG} for the $\AG$-case to simulate the $\AG$-subformulas with the help of $\AU$-formulas as shown in \Cref{fig:au} on page~\pageref{fig:au}. The idea is to introduce another depth proposition after $d_{n+1}$, namely $d_{n+2}$. This is used to express $\AG\phi$ by $\A[\phi\U d_{n+2}]$ without increasing the pathwidth much. \qed
\end{proof}
 
\begin{figure}[h!]
{\scriptsize
 \begin{align*}
   \textit{determined} \dfn & \bigwedge_{i=1}^{n}\Big(q_i\Rightarrow \A[q_i\U d_{n+2}]\Big)\land \bigwedge_{i=1}^{n}\Big(\neg q_i\Rightarrow \A[\neg q_i\U d_{n+2}]\Big)\\
   \textit{depth} \dfn & \bigwedge_{i=0}^{n} \A \Big[\Big((d_i\land \neg d_{i+1}) \Rightarrow
   \big(m_{i \bmod 2} \land \neg m_{1-(i \mod 2)}\\
   & \qquad \qquad \land \A[\neg d_{n+2} \U (d_{i+1}\land \neg d_{i+2} \land \neg d_{n+2})]\Big) \U d_{n+2} \Big]\\
   \textit{setCounter} \dfn & \bigwedge_{i=1}^{n}\A[(d_{i}\land \neg d_{i+1}) \Rightarrow \big((q_{i}\Rightarrow t_{\uparrow part(i)})\land (\neg q_{i}\Rightarrow f_{\uparrow part(i)})\big) \U d_{n+2}]\\
   \textit{incCounter} \dfn & \bigwedge_{p=1}^{k}\bigwedge_{j=0}^{n[p]-1}\bigwedge_{m=0}^{1} \A\Big[\Big((t_{\uparrow p} \land tr^{j}_{p} \land m_i) \Rightarrow \A\big[ m_i \U \A[tr^{j+1}_p \U d_{n+2}]\big]\Big) \land \\
 & \hphantom{\bigwedge_{p=1}^{k}\bigwedge_{j=0}^{n[p]-1}\bigwedge_{m=0}^{1} \A\Big[} \Big((f_{\uparrow p} \land fl^{j}_{p} \land m_i) \Rightarrow \A\big[ m_i \U \A[fl^{j+1}_p \U d_{n+2}]\big]\Big)\Big]\U d_{n+2}\\
\textit{targetMet} \dfn & \bigwedge_{p=1}^{k}\A\Big[\big(d_{n+1} \Rightarrow (tr_{p}^{tg(p)} \land \neg tr_{p}^{tg(p)+1}\land fl_{p}^{n[p]-tg(p)} \land \neg fl_{p}^{n[p]-tg(p)+1} \big)\U d_{n+2} \Big]\\
 \textit{countInit} \dfn & d_0 \land \neg d_1 \land \bigwedge_{p=1}^{k}(\neg tr_{p}^{1} \land \neg fl_{p}^{1} \land \A[(tr_{p}^{0} \land  fl_{p}^{0}) \U d_{n+2}])\\
\textit{countMonotone} \dfn & \bigwedge_{i=1}^{n} \A\left[(d_i \Rightarrow d_{i-1})
\land \bigwedge_{p=1}^{k} \bigwedge_{l=2}^{n[p]} (tr_{p}^{j}\Rightarrow tr_{p}^{j-1}) \land
(fl_{p}^{j}\Rightarrow fl_{p}^{j-1})
 \U d_{n+2}\right]\end{align*}}
\caption{Reduction from p-PW-SAT to $\CTLSAT(\{\AU\})$\label{fig:au}}
\end{figure}

\section{Conclusion} \label{sect:conclusion}

In this work we present an almost complete classification with respect to parameterized complexity of all possible CTL-operator fragments of the satisfiability problem in computation tree logic CTL parameterized by formula pathwidth and temporal depth. Only the case for the fragment containing solely $\AF$ remains open. Currently we are working on a classification of this fragment which aims for an $\FPT$ result and uses the ``full version'' of \Cref{thm:genCourcelle}; the main goal is to bound the model depth of an $\AF$-formula in the full parameter, i.e.,  not only in the temporal depth of the formula. This requires finding lower bounds for the treewidth of the considered structures when the formula enforces a deep model. Then we can construct a family of MSO formulas similar to the $\AX$ case. The classified results form a dichotomy with two fragments in \FPT and the remainder being \W1-hard.

Comparing our results to the situation in usual computational complexity for the decision case they do not behave as expected. Surprisingly the fragment $\{\AX\}$ is \FPT whereas on the decision side this fragment is $\PSPACE$-complete. For the other classified fragments the rule of thumb is the following: The \NP-complete fragments are \FPT whereas the \PSPACE- and \EXPTIME-complete fragments are \W1-hard. For the shown \W1-hardness results an exact classification with matching upper bounds is open for further research. Similarly a complete classification with respect to all possible Boolean fragments in the sense of Post's lattice is one of our next steps.

Furthermore we constructed a generalization of Courcelle's theorem to infinite signatures for parameterized problems $(Q,\kappa)$ with $Q\subseteq\Sigma^*$ such that the treewidth of the relational structures $\mathcal A_x$ corresponding to instances $x\in\Sigma^*$ is $\kappa$-bounded under the existence of a computable family of MSO-formulas (cf.\ \Cref{thm:genCourcelle}). Previously such a general result for infinite signatures was not known to the best of the authors knowledge and is of independent interest.

Another consequent step will be the classification of other temporal logics fragments, e.g., of linear temporal logic LTL and the full branching time logic CTL$^{*}$ with respect to their parameterized complexity. Also the investigation of other parameterizations beyond the usual considered measures of pathwidth or treewidth and temporal depth may lead to a better understanding of intractability in the parameterized sense.

\bibliographystyle{splncs}
\bibliography{pc-ctl}

\end{document}